\documentclass{article} 
\usepackage{hyperref}
\usepackage{url}

\usepackage{amsmath,amssymb,amsfonts,graphicx,amsthm,nicefrac,mathtools,bm}
\usepackage{hyperref}
\usepackage[numbers]{natbib}
\usepackage[english]{babel}
\usepackage{times}
\usepackage[T1]{fontenc}
\usepackage{bbm}
\usepackage{color}
\usepackage{xspace}
\usepackage{rotating,multirow}

\newtheorem{theorem}{Theorem}

\newtheorem{lemma}{Lemma}

\theoremstyle{definition}

\theoremstyle{remark}

\makeatletter
\newtheorem*{rep@theorem}{\rep@title}
\newcommand{\newreptheorem}[2]{%
\newenvironment{rep#1}[1]{%
 \def\rep@title{#2 \ref{##1}}%
 \begin{rep@theorem}}%
 {\end{rep@theorem}}}
\makeatother
\newreptheorem{theorem}{Theorem}
\newreptheorem{lemma}{Lemma}
\newreptheorem{corollary}{Corollary}
\newreptheorem{proposition}{Proposition}

\newcommand{\figref}[1]{Figure~\ref{fig:#1}}

\newcommand{\secref}[1]{Section~\ref{sec:#1}}
\newcommand{\appref}[1]{Appendix~\ref{app:#1}}

\newcommand{\lemref}[1]{Lemma~\ref{lem:#1}}

\newcommand{\thmref}[1]{Theorem~\ref{thm:#1}}

\newcommand{\thmsref}[1]{Theorems~\ref{thm:#1}}
\newcommand{\thmssref}[1]{\ref{thm:#1}}

\newcommand{\eqnref}[1]{\eqref{eqn:#1}}

\DeclareMathOperator{\diag}{diag}

\DeclareMathOperator*{\argmin}{arg\,min}

\newcommand{\one}[1]{{\mathbbm{1}}_{{#1}}}
\newcommand{\inner}[2]{\langle{#1},{#2}\rangle} 

\renewcommand{\O}[1]{\mathcal{O}\left({#1}\right)}
\def\R{\mathbb{R}}

\newcommand{\zeros}{\mathbf{0}}

\newcommand{\ignore}[1]{}

\newcommand{\onebr}[1]{{\mathbbm{1}}\left\{{{#1}}\right\}}
\newcommand{\norm}[1]{\|{#1}\|}
\newcommand{\opnorm}[1]{\|{#1}\|_{\mathsf{op}}}

\usepackage{algorithm,algorithmic}
\usepackage[hmargin=3.25cm,vmargin=3.5cm]{geometry}

\title{The Log-Shift Penalty for Adaptive Estimation
of Multiple Gaussian Graphical Models}

\author{
Yuancheng Zhu and Rina Foygel Barber
\\\small
Department of Statistics, University of Chicago}
\date{}
\begin{document}

\maketitle

\begin{abstract}
Sparse Gaussian graphical models characterize sparse dependence relationships
 between random variables in a network. 
 To estimate multiple related Gaussian graphical models on the same 
set of variables, we formulate a hierarchical model, which leads to an optimization problem
with a nonconvex log-shift penalty function.
We show that under mild conditions the optimization problem is convex despite the inclusion
of a nonconvex penalty, and derive an efficient
optimization algorithm. Experiments on both synthetic and real data show that the proposed 
method is able to achieve good selection and estimation performance simultaneously,
because the nonconvexity of the log-shift penalty allows 
for weak signals to be thresholded to zero
without excessive shrinkage on the strong signals.
\end{abstract}

\section{Introduction and background}
For a set of variables $X_1,\dots,X_p$, a graphical model is commonly used 
to reflect sparse dependence structure among the variables. The presence of an edge $(i,j)$
reflects that variables $X_i$ and $X_j$ are dependent even after controlling for
the effects of the remaining variables.
If $X=(X_1,\dots,X_p)\sim N(\mu,\Sigma)$,
the resulting model is known as a Gaussian graphical
model (GGM), and
in this case the edges (i.e.\ conditional dependencies) correspond to nonzero
entries in the precision matrix, $\Omega=\Sigma^{-1}$. The log-likelihood
for  $\Omega$ after observing $n$ iid draws of $X$ is given by
\[L(\Omega) =  \frac{n}{2}\log\det(\Omega)
   - \frac{n}{2}\inner{S}{\Omega}\;,\]
  where $S$ is the sample covariance of the $n$ iid observations. These types of 
models arise in a wide range of applications, including genetics (modeling interactions
among gene expression levels), finance (finding interactions between different
stock prices), and social networks (modeling relationships among people, spread
of information or disease, etc).

In a high-dimensional setting, where we observe $n$ iid realizations of $X$ for 
a sample size $n<p$, the sparsity of the precision matrix allows us to accurately
estimate the distribution of $X$ even though $\Sigma=\mathsf{Cov}(X)$ is in general
not identifiable from $n<p$ samples. A well-studied convex approach
 to finding $\Omega=\Sigma^{-1}$ is the graphical
  Lasso~\cite{friedman2008sparse}, which calculates\footnote{
  In some works in the literature, $\norm{\Omega}_1=\sum_{ij}|\Omega_{ij}|$ is penalized,
  i.e.\ the diagonal elements are not excluded from the penalty, but we exclude
  them to facilitate comparison with our work.}
  \begin{equation}\label{eqn:glasso}\widehat{\Omega}_{\mathsf{glasso}}=
  \argmin_{\Omega\succeq 0}\Big\{-L(\Omega)
     +\gamma\sum_{i<j}|\Omega_{ij}|\Big\}\;,\end{equation}
 The penalty term
promotes sparsity---due to the shrinkage on the off-diagonal
entries of $\Omega$, many of the $\Omega_{ij}$'s (for $i\neq j$) will be zero when $\gamma$ 
is sufficiently large. Under some conditions, the graphical Lasso is consistent
for edge selection in sparse models, even at sample size $n\ll p$ \cite{ravikumar2011high}.

\paragraph{Multiple graphs} In some applications, we may have multiple sets of observations with related (but
not necessarily identical) covariance structures, for instance when the same variables are measured
across different settings (such as gene expression levels in healthy vs in cancerous
 tissues~\cite{danaher2013joint} or across different phases of an organism's life
 cycle~\cite{kolar2010estimating}). Suppose that we observe data
 from $K$ different GGMs with similar sparsity structures,
 and would like to
  estimate the $K$ precision matrices $\Omega^{(1)},\dots,\Omega^{(K)}$ jointly.
Let   $L_k(\Omega^{(k)})$ be the log-likelihood for the $k$th data set given by
\[L_k(\Omega^{(k)}) = \frac{n_k}{2}\log\det(\Omega^{(k)}) - \frac{n_k}{2}\inner{S^{(k)}}{\Omega^{(k)}}\]
for $k$th sample size $n_k$ and $k$th sample covariance matrix $S^{(k)}$.
   \citet{danaher2013joint} propose the group graphical Lasso:
 \begin{equation}\label{eqn:GGL}\widehat{\mathbf{\Omega}}_{\mathsf{GGL}}=
  \argmin_{\mathbf{\Omega}\in\mathcal{S}_p}\Big\{-\sum_k L_k(\Omega^{(k)})
   + \gamma\sum_{i<j}\left[\nu \norm{\mathbf{\Omega}_{ij}}_1+(1-\nu)\norm{\mathbf{\Omega}_{ij}}_2\right]\Big\}\;,\end{equation}
where
 $\mathcal{S}_p$ is the feasible set of positive semidefinite matrix sequences,
\[\mathcal{S}_p=\left\{\mathbf{\Omega}=(\Omega^{(1)},\dots,\Omega^{(K)})\ : \ \Omega^{(k)}\in\R^{p\times p},\Omega^{(k)}\succeq 0\right\}\;,\]
and $\mathbf{\Omega}_{ij} = (\Omega^{(1)}_{ij},\dots,\Omega^{(K)}_{ij})$ is the vector
of coefficients at position $(i,j)$ across the $K$ settings. If $\nu=1$, the solution $\widehat{\mathbf{\Omega}}_{\mathsf{GGL}}$ reduces to performing a graphical Lasso
on each data set $k=1,\dots,K$;  no information is shared across the $K$ tasks.
 At the other extreme, for $\nu=0$, the penalty $\sum_{i<j}\norm{\mathbf{\Omega}_{ij}}_2$ 
ensures identical sparsity patterns across the $K$ 
estimated precision matrices.

\paragraph{A nonconvex approach} For the single graph setting ($K=1$), recent work by \citet{wong2013adaptive} proposes an adaptive, nonconvex approach, defined by a hierarchical model for each $\Omega_{ij}$:
\begin{align}
\label{eqn:wong_model}\tau_{ij}&\propto 1/\tau_{ij}\text{ for all $i<j$ (an improper prior)},\\
\notag\Omega_{ij}|\tau_{ij}&\sim \mathcal{N} \left(0,\tau_{ij}\right)\text{ for all $i<j$},\\
\notag X|\mu,\Omega&\sim \mathcal N\big(\mu,\Omega^{-1}\big)\text{ iid for each observation}.
\end{align}
This hierarchical model can be viewed as a graphical model version of the Bayesian Lasso 
introduced by \citet{park2008bayesian}. 
Marginalizing over $\tau_{ij}$, this induces a marginal (improper) density $\propto 1/|\Omega_{ij}|^2$, leading
to the MAP estimation problem
\begin{equation}\label{eqn:wong_MAP}
\widehat{\Omega}_{\mathsf{adaptive}}=\argmin_{\Omega\succeq 0} \Big\{-L(\Omega) + \sum_{i<j}\log(|\Omega_{ij}|^2)\Big\}\;. \end{equation}
This procedure is adaptive because, due to the concavity of the log penalty,
 large entries $\Omega_{ij}$ suffer less shrinkage when estimated, as compared to an $\ell_1$-norm penalty
 like the graphical Lasso.
Empirically, 
\citet{wong2013adaptive} find that the adaptive nonconvex penalty leads to improvements
relative to the graphical Lasso \eqnref{glasso}
in terms of accurate estimation and support recovery. However, the optimization
problem \eqnref{wong_MAP} is in general nonconvex and may have local minima.

\paragraph{Contributions} In the work presented here, we formulate a hierarchical model for multiple
graphs, and derive an optimization problem corresponding to finding the  maximum a posteriori (MAP) estimate
for $\mathbf{\Omega}$. The resulting optimization problem combines a likelihood term with
a nonconvex penalty,
leading to reduced shrinkage on edges with strong signals (thus improving over convex-penalty
methods).

 Crucially, even with the nonconvex penalty, our
optimization problem is convex under some mild conditions, thus avoiding
issues with local minima.
Furthermore, we find that the optimization
speedup results of \citet{danaher2013joint} extend to our method. Empirically, our method
is able to simultaneously identify the nonzero edges in a graph (model selection) and estimate the
parameters on these edges---this is a strong advantage of our nonconvex penalty, which
is able to produce a sparse solution while not imposing strong shrinkage on large nonzero estimated values,
while convex-penalty methods generally cannot achieve both at the same tuning parameter value.

\subsection{Outline}
The remainder of this paper is organized as follows. We introduce our method in \secref{method},
which gives a hierarchical model for the $K$ linked GGMs, and derives an objective function to find
the maximum a posteriori (MAP) estimate for $\mathbf{\Omega}=(\Omega^{(1)},\dots,\Omega^{(K)})$.
 We discuss the sparsity and shrinkage
properties of our method, in particular as compared to the group graphical Lasso, in \secref{shrinkage}.
In \secref{optim} we discuss optimization for the objective function defined by our method,
and in particular find conditions that lead to a convex optimization problem that can be split
into smaller subproblems (connected components of the graphs); proofs for the results
in this section can be found in \appref{proofs}. We present experiments
on simulated data and on stock price data in \secref{exper}. Finally, we conclude with a brief
discussion of our work and of future directions in \secref{discuss}.

\section{Methodology}\label{sec:method}

\subsection{A hierarchical model for multiple GGMs}\label{sec:hierarchical}
Consider the following hierarchical models for $K$ Gaussian graphical models with $p$ nodes each:
\begin{align}
\label{eqn:model1}\tau_{ij}&\sim\text{InverseGamma}(\alpha,\beta)\text{ for all $i<j$},\\
\notag\Omega_{ij}^{(k)}|\tau_{ij}&\sim\text{Laplace}\left(\tau_{ij}\right)\text{ for all $k$, for all $i<j$},\\
\notag X^{(k)}|\mu^{(k)},\Omega^{(k)}&\sim \mathcal N\big(\mu^{(k)},(\Omega^{(k)})^{-1}\big)\text{ for all $k$}.
\end{align}
We place a flat prior on $\mu^{(k)}$ and on the diagonal entries $\Omega^{(k)}_{ii}$. Of course, we must require $\Omega^{(k)}\succeq 0$ for each $k$. We may also choose to allow improper priors for $\tau_{ij}$ by allowing $\alpha$ and/or $\beta$ to be zero.

This hierarchical model characterizes
our prior belief regarding shared structure across the $K$ graphs. The common structure across the graphs is governed by the shared parameter $\tau_{ij}$ for the weights on the same edge in different graphs. The hyperparameters $\alpha$ and $\beta$ control the magnitude and the variation
of the $\tau_{ij}$'s and thus the sparsity pattern of the graphs.

\paragraph{Marginal distribution of $\mathbf{\Omega}$ given $\tau$}
We now calculate the marginal prior density of $\mathbf{\Omega}$:
\begin{align}
p(\mathbf{\Omega})
\notag&\propto\one{\mathbf{\Omega}\in\mathcal{S}_p}\cdot \prod_{i<j}\int_{\tau_{ij}} \left[\prod_{k=1}^K p(\Omega_{ij}^{(k)}|\tau_{ij})\right]p(\tau_{ij})\; \mathsf{d}\tau_{ij}\\
\notag&\propto\one{\mathbf{\Omega}\in\mathcal{S}_p}\cdot \prod_{i<j}\int_{\tau_{ij}} \left[\prod_{k=1}^K\tau_{ij}^{-1} e^{-|\Omega^{(k)}_{ij}|/\tau_{ij}}\right]\cdot \tau_{ij}^{-\alpha-1}e^{-\beta/\tau_{ij}}\; \mathsf{d}\tau_{ij}\\
\notag&=\one{\mathbf{\Omega}\in\mathcal{S}_p}\cdot \prod_{i<j}\int_{\tau_{ij}}\tau_{ij}^{-K-\alpha-1} e^{-(\beta+\norm{\mathbf{\Omega}_{ij}}_1)/\tau_{ij}}\; \mathsf{d}\tau_{ij}\\
\label{eqn:omega_marginal}&\propto\one{\mathbf{\Omega}\in\mathcal{S}_p}\cdot\prod_{i<j} (1+\norm{\mathbf{\Omega}_{ij}}_1/\beta)^{-(\alpha+K)}\;,
\end{align}
where the last step is obtained by marginalizing over $\tau_{ij}$
and dividing by the constant $\beta^{-(\alpha+K)}$. When $K=1$, even though our hierarchical model takes a different form than the model \eqnref{wong_model} proposed by \cite{wong2013adaptive}, we obtain the same marginal distribution of $\mathbf{\Omega}$ when we set $\alpha=1$ and $\beta=0$. However, we will show later on that choosing nonzero $\beta$ will allow for a convex optimization problem.

\paragraph{The posterior MAP}
Combining the marginal prior on $\mathbf{\Omega}$ \eqnref{omega_marginal} with the log-likelihoods $L_k(\Omega^{(k)})$ from  the $K$ data sets, we would like to calculate the maximum a posteriori (MAP) estimate:
\begin{equation}\label{eqn:obj}
\widehat{\mathbf{\Omega}}=\argmin_{\mathbf{\Omega}\in\mathcal{S}_p}\Big\{  - \sum_k L_k(\Omega^{(k)})+\gamma\sum_{i<j}\beta\log(1+\|\mathbf{\Omega}_{ij}\|_1/\beta)\Big\}\;,
\end{equation}
where $\gamma=\frac{\alpha+K}{\beta}$ (we introduce this reparametrization for later convenience).
This penalized likelihood function combines a convex negative-log-likelihood term with a nonconvex ``log-shift'' penalty.
While the underlying hierarchical model requires $\gamma\beta\geq K$ by construction, we relax this to $\gamma\geq 0$.

\paragraph{A generalization}
We can also consider replacing $\norm{\mathbf{\Omega}_{ij}}_1$ in \eqnref{obj} with any convex regularizer $f(\mathbf{\Omega}_{ij})$, which leads to the optimization problem
\begin{equation}\label{eqn:obj1}
\widehat{\mathbf{\Omega}}=\argmin_{\mathbf{\Omega}\in\mathcal{S}_p} F(\mathbf{\Omega})\text{\quad where \quad}F(\mathbf{\Omega})\coloneqq -\sum_k L_k(\Omega^{(k)})+\gamma\sum_{i<j} \beta\log(1+f(\mathbf{\Omega}_{ij})/\beta)\;.
\end{equation}
As an important example, we can consider a (sparse) group Lasso penalty on each $\mathbf{\Omega}_{ij}$:
\[f(\mathbf{\Omega}_{ij}) = \nu\norm{\mathbf{\Omega}_{ij}}_1 + (1-\nu)\norm{\mathbf{\Omega}_{ij}}_2\;.\]
In fact, the penalized likelihood optimization problem in \eqnref{obj1} can be motivated by a generalization of our hierarchical model from \secref{hierarchical}. If $f(\mathbf{\Omega}_{ij})=\norm{\mathbf{\Omega}_{ij}}$ for some norm $\norm{\cdot}$, take
\begin{align}
\label{eqn:model2}\tau_{ij}&\sim\text{InverseGamma}(\alpha,\beta)\text{ for all $i<j$},\\
\notag\mathbf{\Omega}_{ij}|\tau_{ij}&\propto \tau_{ij}^{-K}e^{-\norm{\mathbf{\Omega}_{ij}}/\tau_{ij}}\text{ for all $i<j$},\\
\notag X^{(k)}|\Omega^{(k)}&\sim  \mathcal N\big(\mu^{(k)},(\Omega^{(k)})^{-1}\big)\text{ for all $k$}.
\end{align}
Next, marginalizing over $\tau$ (following similar steps as in \eqnref{omega_marginal} earlier),
\[p(\mathbf{\Omega})\propto \one{\mathbf{\Omega}\in\mathcal{S}_p}\cdot \prod_{i<j} (1+\norm{\mathbf{\Omega}_{ij}}/\beta)^{-(\alpha+K)}\;.\]
Combining this with the likelihood terms yields the penalized optimization problem \eqnref{obj1}.

\subsection{Sparsity and shrinkage of $\mathbf{\Omega}$}\label{sec:shrinkage}
We next examine the effects of the parameters $\gamma$ and $\beta$ in the log-shift objective function \eqnref{obj1}, which arise from $\alpha$ and $\beta$ in the hierarchical model \eqnref{model2}. To understand their role in inducing sparsity and shrinkage in $\mathbf{\Omega}$, we first consider the function
\[g_\beta(x) =  \beta \log(1+|x|/\beta)\;.\]
In a sparse regression setting, this type of penalty function has been studied by  \citet{candes2008enhancing} and others in the context of reweighted $\ell_1$ minimization, and was found to preserve the desirable sparsity properties of $\ell_1$ regularization while reducing the amount of shrinkage on large coefficients. 

The penalty function $g_\beta(x)$ behaves like a $\ell_1$ penalty when $|x|/\beta\approx 0$, which we can see by taking a local linear approximation to the log function:
\[\log(1+|x|/\beta)\approx |x|/\beta \ \Rightarrow \ g_\beta(x)\approx |x|\;.\]
On the other hand, as $|x|/\beta$ grows large, the concavity of the log function becomes apparent, and therefore there is less shrinkage on large values of $x$. See \figref{nonconvex} for an illustration.
\begin{figure}
\centering
\includegraphics[width=3.5in]{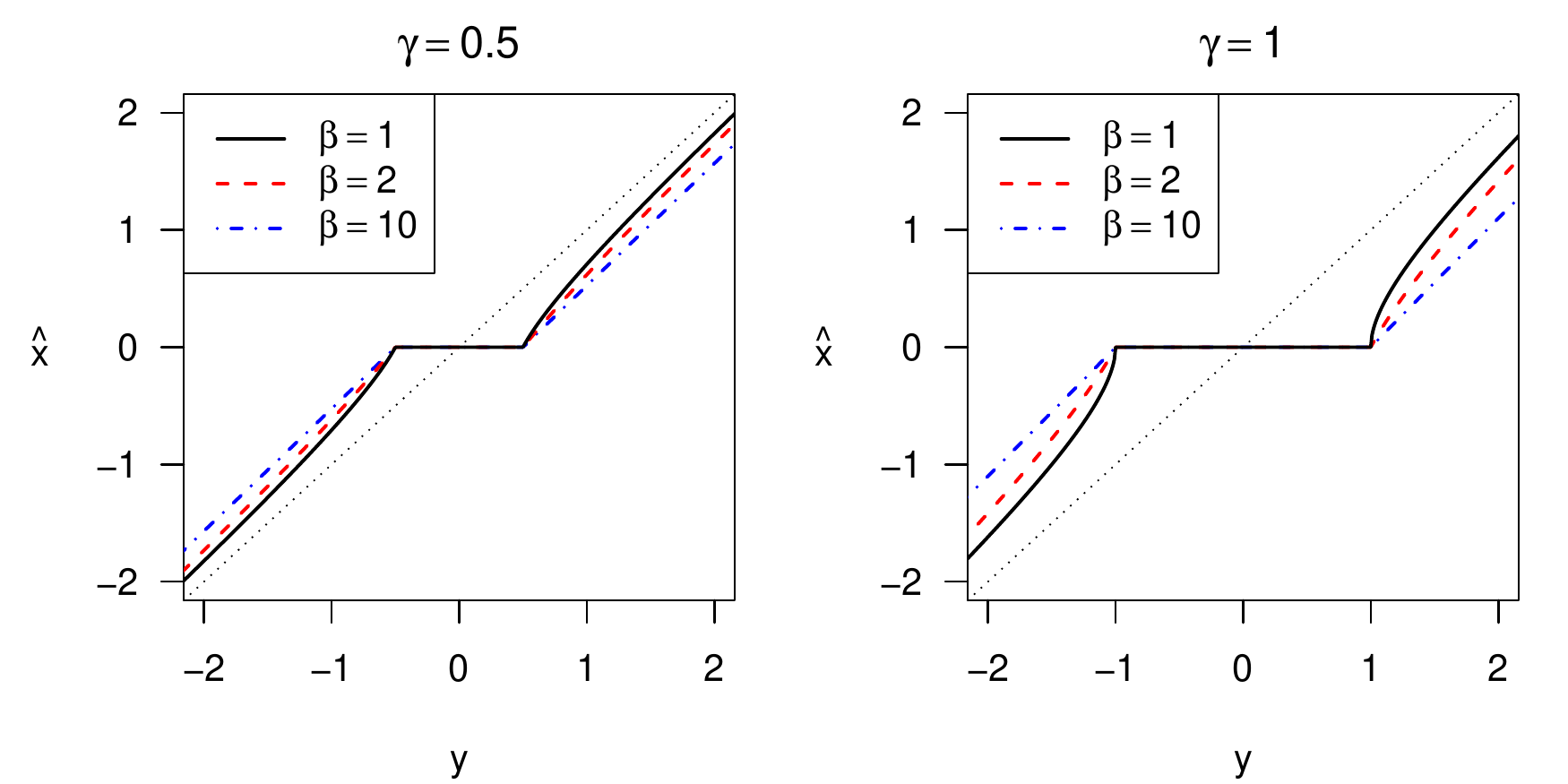}
\caption{Sparsity and shrinkage behavior of the solution $\hat{x}=\argmin\left\{\frac{1}{2}(y-x)^2 + \gamma g_\beta(x)\right\}$. Note that $\gamma$ affects the point at which the solution is thresholded to zero, while $\beta$ controls the nonconvexity (and therefore, the shrinkage) for nonzero solutions.}
\label{fig:nonconvex}
\end{figure}

Next, we return to our prior distribution on $\mathbf{\Omega}$. Comparing the penalized likelihood function \eqnref{obj1} with our calculations with $g_\beta(\cdot)$ above, we can interpret the parameters $\beta$ and $\gamma$ in \eqnref{obj1} as follows:
\begin{itemize}
\item $\gamma$ controls the amount of penalization on $\mathbf{\Omega}$, and thus the sparsity level of the solution.
\item $\beta$ controls the nonconvexity of the penalty, with small $\beta$ yielding reduced shrinkage in the estimate of $\mathbf{\Omega}$ (but possible nonconvexity of the objective function), while $\beta\rightarrow\infty$ causes the penalty term to approach $\gamma\sum_{i<j} f(\mathbf{\Omega}_{ij})$. 
\end{itemize}

\subsection{Other related work}
The graphical Lasso \cite{friedman2008sparse} and group graphical Lasso \cite{danaher2013joint} methods, discussed above in \eqnref{glasso} and \eqnref{GGL}, both propose estimation of $\mathbf{\Omega}$ via a convex penalty. Our method may be viewed as a generalization of the group graphical Lasso, which is obtained by setting $\beta=\infty$ in our log-shift penalty.

Turning to nonconvex methods, in addition to \citet{wong2013adaptive}'s model described in \eqnref{wong_model} above, we are aware of several other methods using nonconvex regularization, all of which allow for reduced shrinkage on large entries, but may potentially lead to nonconvex optimization problems. First, for estimation of a single graph, \citet{fan2009network} apply an adaptive Lasso (reweighted $\ell_1$) penalty to the GGM setting, optimizing
\[\argmin_{\Omega\succeq 0}\Big\{-L(\Omega) + \gamma\sum_{i<j} {|\Omega_{ij}|}/{|\tilde{\Omega}_{ij}|^\alpha}\Big\}\;,\]
where $\tilde{\Omega}$ is some initial estimate of $\Omega$. In fact, for $\alpha=1$, this reweighted $\ell_1$ penalty is can be viewed as a single iteration towards solving \citet{wong2013adaptive}'s MAP estimation problem \eqnref{wong_MAP}, although this is not the  approach taken in \cite{wong2013adaptive} (see  \cite{candes2008enhancing} for the sparse regression setting). \citet{fan2009network} also examine a SCAD penalty on each $|\Omega_{ij}|$, which behaves similarly.

Finally, in the multiple graph setting, \citet{guo2011joint} propose the optimization problem
 \begin{equation}\label{eqn:guo}\widehat{\mathbf{\Omega}}_{\mathsf{sqrt}}=
  \argmin_{\mathbf{\Omega}\in\mathcal{S}_p}\Big\{-\sum_k L_k(\Omega^{(k)})
   + \gamma\sum_{i<j}\sqrt{\norm{\mathbf{\Omega}_{ij}}_1}\Big\}\;.\end{equation}
This nonconvex penalty
encourages similar sparsity patterns in the $K$ graphs, but does not allow for tuning the amount of nonconvexity or the balance between shared support vs different support.

\section{Convexity and optimization}\label{sec:optim}
In this section, we derive a simple condition on the parameters $\beta$ and $\gamma$ in our log-shift method \eqnref{obj1} that guarantees the convexity of the objective function $F(\mathbf{\Omega})$ over a bounded set. We then develop a majorization-minimization algorithm for finding the global minimum, and a preprocessing step where the graphs are split into connected components, allowing for smaller optimization problems that can be solved in parallel.
While we are primarily interested in regularizers of the form $f(\cdot)=\nu\norm{\cdot}_1+(1-\nu)\norm{\cdot}_2$, our results apply more generally to any convex regularizer $f(\cdot)$.

\subsection{Convexity}
To ensure that $F(\mathbf{\Omega})$ is convex, we will place a bound on $\mathbf{\Omega}$  to obtain strong convexity of the likelihood term, while placing a lower bound on $\beta$ to control the nonconvexity of the penalty term. The result below may be viewed as an application of \citet{loh2013regularized}'s results on nonconvex regularizers.

The condition that we require on $\mathbf{\Omega}$ is mild.  For any $b=(b_1,\dots,b_K)\in\R^K_+$, define
\[\mathcal{S}_p(b)=\left\{\mathbf{\Omega}\in\mathcal{S}_p \ : \ \opnorm{\Omega^{(k)}}\leq b_k\text{ for }k=1,\dots,K\right\}\;,\]
where $\opnorm{\cdot}$ is the matrix operator norm (largest singular value). This is a reasonable nondegeneracy condition on the $K$ graphical models underlying the data.

\begin{theorem}\label{thm:obj_fn_convex}
If $f(\cdot)$ is convex, nonnegative, and $L$-Lipschitz, and if
\begin{equation}\label{eqn:convex_condition}\beta\geq \frac{\gamma L^2}{2} \cdot \max_k \frac{b_k^2}{n_k}\;,\end{equation}
then
 $F(\mathbf{\Omega})$ is convex over $\mathbf{\Omega}\in\mathcal{S}_p(b)$. If \eqnref{convex_condition} is satisfied with a strict inequality, then we obtain strict convexity.
\end{theorem}
We note that this allows for a value of $\beta$ that is very small if the sample sizes $n_k$ are large, that is, even if the penalty is highly nonconvex, as desired to avoid excessive shrinkage on strong signals. The proof of this theorem, given in \appref{proofs}, simply shows that the strong convexity of the likelihood term in $F(\mathbf{\Omega})$ is sufficient to counterbalance the concavity of the log penalty. 

\subsection{Optimization via majorization-minimization}\label{sec:majmin}
To minimize $F(\mathbf{\Omega})$ we use majorization-minimization \cite{hunter2004tutorial}. Let $\tilde{\mathbf{\Omega}}$ be our current estimate of $\widehat{\mathbf{\Omega}}$. Since $\log(\cdot)$ is concave, we   bound $\log(1+f(\mathbf{\Omega}_{ij})/\beta)$ by the linear approximation centered at $\tilde{\mathbf{\Omega}}_{ij}$:
\[\log\big(1+f(\mathbf{\Omega}_{ij})/\beta\big)\leq \log\big(1+f(\tilde{\mathbf{\Omega}}_{ij})/\beta\big) + \frac{f(\mathbf{\Omega}_{ij})/\beta - f(\tilde{\mathbf{\Omega}}_{ij})/\beta}{1+f(\tilde{\mathbf{\Omega}}_{ij})/\beta}\;.\]
Then the objective function is bounded as
\[F(\mathbf{\Omega}) \leq  \underbrace{-\sum_k L_k(\Omega^{(k)})+ \gamma\sum_{i<j} \frac{ f(\mathbf{\Omega}_{ij}) }{1+f(\tilde{\mathbf{\Omega}}_{ij})/\beta}}_{\eqqcolon F(\mathbf{\Omega};\tilde{\mathbf{\Omega}})} + \left(\text{\begin{tabular}{c}terms that are\\ constant with\\ respect to $\mathbf{\Omega}$\end{tabular}}\right)\;,\]
with equality at $\mathbf{\Omega}=\tilde{\mathbf{\Omega}}$. Note that $F(\mathbf{\Omega};\tilde{\mathbf{\Omega}})$ is a convex function of $\mathbf{\Omega}$. Therefore, to find $\widehat{\mathbf{\Omega}}$,
\begin{enumerate}
\item Initialize $\mathbf{\Omega}_{[0]} =(\zeros_{p\times p},\dots,\zeros_{p\times p})$  (or any other initial value).
\item For $t=1,2,\dots$, solve the convex optimization problem 
\begin{equation}\label{eqn:majmin_step}\mathbf{\Omega}_{[t]} = \argmin_{\mathbf{\Omega}\in \mathcal{S}_p(b)} F(\mathbf{\Omega}; \mathbf{\Omega}_{[t-1]})\;.\end{equation}
\item Stop when some convergence criterion has been reached.
\end{enumerate}

For optimizing \eqnref{majmin_step}, if $f(\cdot)$ is chosen to be the sparse group Lasso regularizer
\[f(\cdot)=\nu\norm{\cdot}_1+(1-\nu)\norm{\cdot}_2\;,\]
 then the step \eqnref{majmin_step} is equivalent to a weighted group graphical Lasso problem \cite{danaher2013joint}, but with an additional constraint that $\mathbf{\Omega}_{[t]}\in\mathcal{S}_p(b)$; this constraint can be added to the ADMM algorithm for group graphical Lasso given in \cite{danaher2013joint} with no additional computational cost.

If the objective function $F(\mathbf{\Omega})$ is convex over $\mathcal{S}_p(b)$---that is, if our choices of $\beta$, $\gamma$, and $b$ satisfy the condition \eqnref{convex_condition} of \thmref{obj_fn_convex}---then majorization-minimization is guaranteed to converge to a globally optimal solution $\widehat{\mathbf{\Omega}}\in\argmin_{\mathbf{\Omega}\in\mathcal{S}_p(b)}F(\mathbf{\Omega})$ \cite{wu1983convergence}. In practice, we may choose to remove the bound on the spectral norms, or equivalently, explore concavity of the penalty beyond what is allowed in the convexity condition \eqnref{convex_condition}, since lower values of $\beta$ may perform better empirically.

\subsection{Separation into connected components}
For the graphical Lasso \eqnref{glasso}, \citet{witten2011new} proved that the connected components of the solution $\widehat{{\Omega}}_{\mathsf{glasso}}$ can be identified in a preprocessing step that simply requires screening for sample correlations $S_{ij}$ that exceed the penalty parameter value $\gamma$. This allows for significantly faster optimization of the graphical Lasso. Theorem 2 of \citet{danaher2013joint} extends this result to the group graphical Lasso setting, by screening for any $i<j$ such that
\begin{equation}\label{eqn:block_GGL}\sqrt{\sum_k \left(n_k |S^{(k)}_{ij}|-\gamma\nu\right)_+^2}>\gamma(1-\nu)\end{equation}
and then solving separate optimization problems for each resulting connected component. Their results prove that the combined solution is a global minimizer of the group graphical Lasso \eqnref{GGL}.

This type of block-wise optimization can be extended to the nonconvex log-shift penalty:
\begin{theorem}\label{thm:block_sep}
Consider any partition $\mathcal{A}=\{A_1,\dots,A_m\}$ of $[p]$  into disjoint sets. Suppose that 
\begin{equation}\label{eqn:block_condition}-\gamma^{-1}\cdot \diag\{n_1,\dots,n_K\}\cdot \mathbf{S}_{ij}\in\partial f(\zeros)\text{ for all $i\not\sim_{\mathcal{A}} j$}\end{equation}
where $\mathbf{S}_{ij}=(S_{ij}^{(1)},\dots,S_{ij}^{(K)})$. If the conditions of \thmref{obj_fn_convex} are satisfied, then there exists some $\widehat{\mathbf{\Omega}}\in\argmin_{\mathbf{\Omega}\in\mathcal{S}_p(b)}F(\mathbf{\Omega})$ such that
$\widehat{\Omega}^{(k)}_{ij}=0$  for all $k$ and all $i\not\sim_{\mathcal{A}}j$. 
\end{theorem}
In particular, if $f(\cdot)=\nu\norm{\cdot}_1+(1-\nu)\norm{\cdot}_2$, then condition \eqnref{block_condition} is equivalent to \citet{danaher2013joint}'s condition \eqnref{block_GGL} for the group graphical Lasso.  Although our penalty is nonconvex, near zero it is approximately equal to the group graphical Lasso penalty (or, more generally, $\beta\log\big(1+f(\mathbf{\Omega}_{ij})/\beta\big)\approx f(\mathbf{\Omega}_{ij})$ when $\mathbf{\Omega}_{ij}\approx 0$). This allows us to extend the proof techniques of \cite{danaher2013joint} to this nonconvex penalty setting.
\thmref{block_sep} is proved in \appref{proofs}.

Based on this result, we now propose a faster algorithm for minimizing $F(\mathbf{\Omega})$.
\begin{enumerate}
\item Partition $[p]$ into sets $A_1,\dots,A_M$, the connected components of the adjacency matrix $C$:
\[C_{ij} = \onebr{-\gamma^{-1}\cdot \diag\{n_1,\dots,n_K\}\cdot \mathbf{S}_{ij}\in\partial f(\zeros)} \text{ for each }i<j\;.\]
\item For  $m=1,\dots,M$, use majorization-minimization (\secref{majmin}) to solve the $m$th block,
\[\widehat{\mathbf{\Omega}}_m = \argmin_{\mathbf{\Omega}\in\mathcal{S}_{p_m}(b)} F_m(\mathbf{\Omega})\text{\quad where $p_m=|A_m|$ and}\]
\[F_m(\mathbf{\Omega})=\sum_{k=1}^K\left[-\frac{n_k}{2}\log\det(\Omega^{(k)})+\frac{n_k}{2}\langle\Omega^{(k)},S^{(k)}_{A_m,A_m}\rangle\right]+\gamma\!\!\sum_{\substack{i<j\\ i,j\in A_m}} \!\!\beta\log(1+f(\mathbf{\Omega}_{ij})/\beta)\;.\]
\item $\widehat{\mathbf{\Omega}}$ concatenates the blocks:  $\widehat{\mathbf{\Omega}}_{A_m,A_m}=\widehat{\mathbf{\Omega}}_m$ for all $m$, and $\widehat{\mathbf{\Omega}}_{A_m,A_{m'}}=\zeros$ for all $m\neq m'$.
\end{enumerate}
If the convexity condition \eqnref{convex_condition} is satisfied, then \thmsref{obj_fn_convex} and \thmssref{block_sep} guarantee that the resulting solution $\widehat{\mathbf{\Omega}}$ is a global minimizer of $F(\mathbf{\Omega})$ over the set $\mathcal{S}_p(b)$.

 \section{Experiments}\label{sec:exper}

\subsection{Simulations}

\paragraph{Data and methods} We simulate $K=3$ tridiagonal precision matrices of dimension $p=100$, following the autoregressive (AR)
process example in \citet{fan2009network}. The precision matrices $\Omega^{(k)}$ have identical tridiagonal support, but have different nonzero values (see example 4.1 in \cite{fan2009network} for details).
For each $k$, we draw $n_k=40$ iid samples from the distribution $\mathcal N\big(0,(\Omega^{(k)})^{-1}\big)$.

 To implement our proposed method, we take
 $f(\cdot)=\nu\norm{\cdot}_1+(1-\nu)\norm{\cdot}_2$ and minimize the objective function \eqnref{obj1}
for different values of tuning parameters $(\gamma,\beta,\nu)$. We also test the group graphical Lasso \cite{danaher2013joint}, the graphical Lasso \cite{friedman2008sparse}, and \citet{guo2011joint}'s square-root method, for comparison.
\footnote{Computations for simulations and for the real data experiment were performed in R \cite{R} and used the \texttt{glasso} \cite{R_glasso} and \texttt{JGL} \cite{R_JGL} packages. Code for \citet{guo2011joint}'s method was obtained from the online supplementary material for \cite{danaher2013joint}, available at \url{http://onlinelibrary.wiley.com/journal/10.1111/(ISSN)1467-9868/homepage/76_2.htm}.}

\paragraph{Results} \figref{simulation}(a) displays results from 10 trials. 
Among the methods considered, 
the log-shift method attains the lowest error in recovering $\mathbf{\Omega}$
and simultaneously selects an appropriately low number of edges, when tuning
parameters are chosen judiciously (typically with a lower value of $\beta$, i.e.\ with 
high nonconvexity in the penalty). 
As $\beta$ increases, the performance of our method approaches that of the group graphical Lasso. 
In order to select appropriate tuning parameters for each method
in a data-driven way, we generate a validation data set of same size 
and compute the log likelihoods of the validation data using the estimated precision matrices. 
For our model, the estimate selected by the validation score achieves the minimum error measure,
and yields a total number of selected edges that is close to the number of true edges.

\begin{figure}[t]
\centering
\includegraphics[width=5.7in]{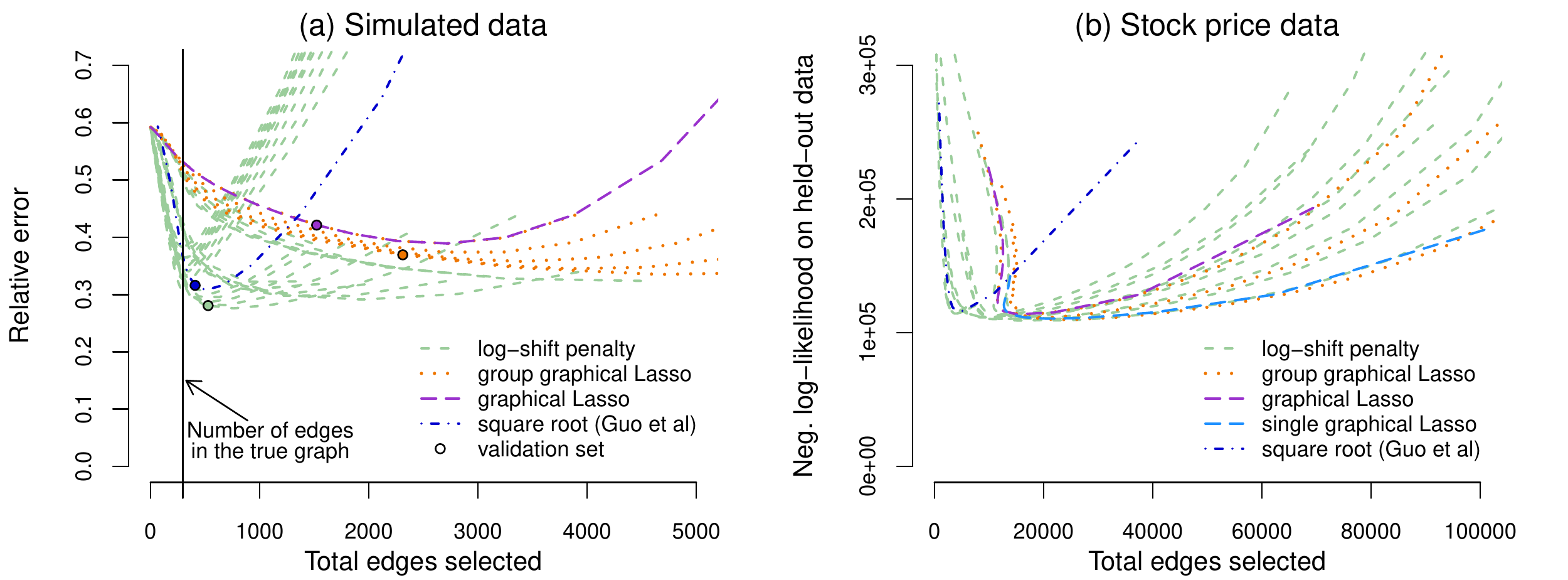}
\caption{Experiment results for estimating precision matrices based on simulated data and stock price data.
For each method, each line represents estimates with various values of $\gamma$ while fixing
other tuning parameters.
(a) Simulated data: relative error in estimating $\mathbf{\Omega}$
versus the total number of edges selected.  For each model, a held-out validation
set was used to select tuning parameter value(s), highlighted in the plot.
(b) Stock price data: negative log likelihood on held-out data versus the total number of edges selected.
Plot is best viewed in color.}
\label{fig:simulation}
\end{figure}

\subsection{Stock price data}

\paragraph{Data and methods} We next test our method on stock price data from Yahoo!~Finance.\footnote{Data available at \url{finance.yahoo.com}} We collected the daily closing prices for $p=432$ stocks that were consistently in the S\&P 500 index from January 1, 2003 to December 31, 2012. Let $S_{i,j}$ be the closing price for stock $j$ on day $i$, and $X_{i,j}=\log (S_{i,j}/S_{i-1,j})$ be the log return. We marginally transform the log returns of each stock to a normal distribution. Denoting the transformed data still as $X_{i,j}$, we treat $X_{i,\cdot}$ as independent observations, although they in fact form a time series. We divide the data into two time periods, one for before (2003--2007) and one for after (2009-2012) the 2008 financial crisis, and remove the data from 2008. The two sample sizes are $n_{\text{before}}=1257$ and $n_{\text{after}}=1005$. In the belief that the relationship between stocks might have changed during the financial crisis, we model the data as two GGMs with similar but non-identical precision matrices. We select 20\% of the data in each time period as training data, and hold out the remaining 80\% to evaluate the performances. We implement our log-shift method with $K=2$, and compare to the same existing methods as before. For an additional comparison, we also fit a single graphical Lasso to the combined data set.
 
\paragraph{Results} To evaluate the results, we calculate the likelihood of the held-out data under the fitted models for each method. Results are displayed in \figref{simulation}(b). While the various methods' best scores are similar for the log-shift, group graphical Lasso, and single graphical Lasso, the log-shift method is able to attain this best validation score with a substantially smaller number of selected edges relative to the convex methods, demonstrating the benefit of the nonconvex penalty.

\section{Discussion}\label{sec:discuss}

In this paper, we introduce a family of nonconvex penalty functions, called the log-shift function, for estimating multiple
related GGMs. It arises from a simple hierarchical model and generalizes existing methods for learning multiple GGMs,
such as the group graphical Lasso \cite{danaher2013joint}.
Compared with methods that use a convex penalty function, the nonconvexity of the penalty function leads to less bias on strong signals
and thus makes it possible to obtain good selection and estimation result at the same time. The log-shift penalty can also
be applied to estimating other models, such as undirected graphical models with non-Gaussian distributions, time-varying GGMs, etc., which we leave to future work.

\bibliographystyle{plainnat}
\bibliography{log_shift_ggm_arxiv}

\begin{thebibliography}{16}
\providecommand{\natexlab}[1]{#1}
\providecommand{\url}[1]{\texttt{#1}}
\expandafter\ifx\csname urlstyle\endcsname\relax
  \providecommand{\doi}[1]{doi: #1}\else
  \providecommand{\doi}{doi: \begingroup \urlstyle{rm}\Url}\fi

\bibitem[Candes et~al.(2008)Candes, Wakin, and Boyd]{candes2008enhancing}
Emmanuel~J Candes, Michael~B Wakin, and Stephen~P Boyd.
\newblock Enhancing sparsity by reweighted $\ell_1$ minimization.
\newblock \emph{Journal of Fourier analysis and applications}, 14\penalty0
  (5-6):\penalty0 877--905, 2008.

\bibitem[Danaher(2013)]{R_JGL}
Patrick Danaher.
\newblock \emph{JGL: Performs the Joint Graphical Lasso for sparse inverse
  covariance estimation on multiple classes}, 2013.
\newblock URL \url{http://CRAN.R-project.org/package=JGL}.
\newblock R package version 2.3.

\bibitem[Danaher et~al.(2013)Danaher, Wang, and Witten]{danaher2013joint}
Patrick Danaher, Pei Wang, and Daniela~M Witten.
\newblock The joint graphical {L}asso for inverse covariance estimation across
  multiple classes.
\newblock \emph{Journal of the Royal Statistical Society: Series B (Statistical
  Methodology)}, 2013.

\bibitem[Fan et~al.(2009)Fan, Feng, and Wu]{fan2009network}
Jianqing Fan, Yang Feng, and Yichao Wu.
\newblock Network exploration via the adaptive {L}asso and {SCAD} penalties.
\newblock \emph{The annals of applied statistics}, 3\penalty0 (2):\penalty0
  521, 2009.

\bibitem[Friedman et~al.(2008)Friedman, Hastie, and
  Tibshirani]{friedman2008sparse}
Jerome Friedman, Trevor Hastie, and Robert Tibshirani.
\newblock Sparse inverse covariance estimation with the graphical {L}asso.
\newblock \emph{Biostatistics}, 9\penalty0 (3):\penalty0 432--441, 2008.

\bibitem[Friedman et~al.(2011)Friedman, Hastie, and Tibshirani]{R_glasso}
Jerome Friedman, Trevor Hastie, and Rob Tibshirani.
\newblock \emph{glasso: Graphical {L}asso- estimation of Gaussian graphical
  models}, 2011.
\newblock URL \url{http://CRAN.R-project.org/package=glasso}.
\newblock R package version 1.7.

\bibitem[Guo et~al.(2011)Guo, Levina, Michailidis, and Zhu]{guo2011joint}
Jian Guo, Elizaveta Levina, George Michailidis, and Ji~Zhu.
\newblock Joint estimation of multiple graphical models.
\newblock \emph{Biometrika}, 98\penalty0 (1):\penalty0 1--15, 2011.

\bibitem[Hunter and Lange(2004)]{hunter2004tutorial}
David~R Hunter and Kenneth Lange.
\newblock A tutorial on {M}{M} algorithms.
\newblock \emph{The American Statistician}, 58\penalty0 (1):\penalty0 30--37,
  2004.

\bibitem[Kolar et~al.(2010)Kolar, Song, Ahmed, Xing,
  et~al.]{kolar2010estimating}
Mladen Kolar, Le~Song, Amr Ahmed, Eric~P Xing, et~al.
\newblock Estimating time-varying networks.
\newblock \emph{The Annals of Applied Statistics}, 4\penalty0 (1):\penalty0
  94--123, 2010.

\bibitem[Loh and Wainwright(2013)]{loh2013regularized}
Po-Ling Loh and Martin~J Wainwright.
\newblock Regularized {M}-estimators with nonconvexity: Statistical and
  algorithmic theory for local optima.
\newblock In \emph{Advances in Neural Information Processing Systems}, pages
  476--484, 2013.

\bibitem[Park and Casella(2008)]{park2008bayesian}
Trevor Park and George Casella.
\newblock The {B}ayesian {L}asso.
\newblock \emph{Journal of the American Statistical Association}, 103\penalty0
  (482):\penalty0 681--686, 2008.

\bibitem[{R Core Team}(2013)]{R}
{R Core Team}.
\newblock \emph{R: A Language and Environment for Statistical Computing}.
\newblock R Foundation for Statistical Computing, Vienna, Austria, 2013.
\newblock URL \url{http://www.R-project.org/}.

\bibitem[Ravikumar et~al.(2011)Ravikumar, Wainwright, Raskutti, Yu,
  et~al.]{ravikumar2011high}
Pradeep Ravikumar, Martin~J Wainwright, Garvesh Raskutti, Bin Yu, et~al.
\newblock High-dimensional covariance estimation by minimizing
  $\ell_1$-penalized log-determinant divergence.
\newblock \emph{Electronic Journal of Statistics}, 5:\penalty0 935--980, 2011.

\bibitem[Witten et~al.(2011)Witten, Friedman, and Simon]{witten2011new}
Daniela~M Witten, Jerome~H Friedman, and Noah Simon.
\newblock New insights and faster computations for the graphical {L}asso.
\newblock \emph{Journal of Computational and Graphical Statistics}, 20\penalty0
  (4):\penalty0 892--900, 2011.

\bibitem[Wong et~al.(2013)Wong, Awate, and Fletcher]{wong2013adaptive}
Eleanor Wong, Suyash Awate, and P~Thomas Fletcher.
\newblock Adaptive sparsity in {G}aussian graphical models.
\newblock In \emph{Proceedings of The 30th International Conference on Machine
  Learning}, pages 311--319, 2013.

\bibitem[Wu(1983)]{wu1983convergence}
CF~Jeff Wu.
\newblock On the convergence properties of the {E}{M} algorithm.
\newblock \emph{The Annals of statistics}, pages 95--103, 1983.

\end{thebibliography}

\appendix
\section{Proofs}\label{app:proofs}
\begin{proof}[Proof of Theorem 1]
For any $\Omega^{(k)}\succeq 0$ with $\opnorm{\Omega^{(k)}}\leq b_k$,
\begin{multline*}\lambda_{\min}\left[\nabla^2_{\Omega^{(k)}}\left(\inner{\Omega^{(k)}}{S^{(k)}}-\log\det(\Omega^{(k)})\right) \right]=\\\lambda_{\min}\left[\Omega^{(k)}{}^{-1}\otimes \Omega^{(k)}{}^{-1}\right]=\left[\lambda_{\min}(\Omega^{(k)}{}^{-1})\right]^2 = \frac{1}{\opnorm{\Omega^{(k)}}^2}\geq\frac{1}{b_k^2}\;.\end{multline*}
Then
\[\mathbf{\Omega}\mapsto \sum_{k=1}^K \frac{n_k}{2}\left(\inner{\Omega^{(k)}}{S^{(k)}}-\log\det(\Omega^{(k)})-\frac{\norm{\Omega^{(k)}}^2_F}{b_k^2}\right)\]
is a convex function over $\mathbf{\Omega}\in\mathcal{S}_p(b)$. Applying \lemref{log_is_convex} given below, the function
\[x\mapsto \log(1+f(x)/\beta) + \frac{L^2}{2\beta^2}\norm{x}^2_2\]
 is convex, and so the following is a convex function over $\mathbf{\Omega}\in\mathcal{S}_p(b)$:
\begin{align*}
\mathbf{\Omega}&\mapsto \sum_{k=1}^K \frac{n_k}{2}\left(\inner{\Omega^{(k)}}{S^{(k)}}-\log\det(\Omega^{(k)})-\frac{\norm{\Omega^{(k)}}^2_F}{b_k^2}\right)+ \gamma\sum_{i<j}\beta\left[\log\big(1+f(\mathbf{\Omega}_{ij})/\beta\big)+\frac{L^2}{2\beta^2}\norm{\mathbf{\Omega}_{ij}}^2_2\right]\\
&= F(\mathbf{\Omega}) - \sum_{k=1}^K \frac{n_k}{2}\cdot \frac{\norm{\Omega^{(k)}}^2_F}{b_k^2}+ \gamma\sum_{i<j}\frac{L^2}{2\beta}\norm{\mathbf{\Omega}_{ij}}^2_2\\
&= F(\mathbf{\Omega}) - \sum_{k=1}^K\norm{\Omega^{(k)}}^2_F\cdot \left( \frac{n_k}{2b_k^2} - \frac{\gamma L^2}{4\beta}\right)\;,
\end{align*}
where the switch from a $2$ to a $4$ in the last step comes from the fact that $\mathbf{\Omega}_{ij}$ is penalized for $i<j$ but not $i>j$.

This proves that $F(\mathbf{\Omega})$ is convex over $\mathcal{S}_p(b)$ as long as $ \frac{n_k}{2b_k^2} \geq \frac{\gamma L^2}{4\beta}$ for all $k$, which is equivalent to the condition in the theorem. If this inequality is strictly satisfied, then this implies strict convexity of $F(\mathbf{\Omega})$.
\end{proof}

\begin{lemma}\label{lem:log_is_convex}
Let $f:\R^p\rightarrow\R$ be a $L$-Lipschitz convex nonnegative function and fix any $\beta>0$. Then
\[x\mapsto \log\big(1+f(x)/\beta\big) + \frac{L^2}{2\beta^2}\norm{x}^2_2\]
is a convex function.
\end{lemma}
\begin{proof}
Take any $x,y\in\R^p$, and any $t\in[0,1]$. 
 Then, 
using the convexity of $f(\cdot)$,
\begin{align*}
&\log\big(1 + f(t\cdot x + (1-t)\cdot y)/\beta\big)\\
&\leq \log\big(1+t\cdot f(x)/\beta+(1-t)\cdot f(y)/\beta\big)\\
&= \log\big(t\cdot(1+f(x)/\beta) + (1-t)\cdot(1+f(y)/\beta)\big)\\
\intertext{Since $\frac{\partial^2}{\partial z^2}\log(z)\in [-1,0]$ for all $z\geq1$,}
&\leq t\cdot \log\big(1+f(x)/\beta\big) + (1-t)\cdot\log\big(1+f(y)/\beta\big) + \frac{t(1-t)}{2\beta^2}\cdot (f(y)-f(x))^2\\
&\leq t\cdot \log\big(1+f(x)/\beta\big) + (1-t)\cdot\log\big(1+f(y)/\beta\big) + \frac{t(1-t)}{2\beta^2}\cdot L^2\norm{x-y}^2_2\;.
\end{align*}
Then
\begin{align*}
&\log\big(1 + f(t\cdot x + (1-t)\cdot y)/\beta\big) + \frac{L^2}{2\beta^2}\norm{t\cdot x + (1-t)\cdot y}^2_2\\
&\leq t\cdot \log\big(1+f(x)/\beta\big) + (1-t)\cdot\log\big(1+f(y)/\beta\big)\\
&\hspace{1in} + \frac{t(1-t)}{2\beta^2}\cdot L^2\norm{x-y}^2_2+ \frac{L^2}{2\beta^2}\norm{t\cdot x + (1-t)\cdot y}^2_2\\
&\leq t\cdot \left[\log\big(1+f(x)/\beta\big)+\frac{L^2}{2\beta^2}\norm{x}^2_2\right] + (1-t)\cdot\left[\log\big(1+f(y)/\beta\big) +\frac{L^2}{2\beta^2}\norm{y}^2_2\right]\;,
\end{align*}
proving convexity of the function as desired.
\end{proof}

\begin{proof}[Proof of Theorem 2]
Define
\[\mathcal{S}_p(b;\mathcal{A})=\left\{\mathbf{\Omega}\in\mathcal{S}_p(b):\Omega^{(k)}_{ij}=0\text{ for all $k$ and all $i\not\sim_{\mathcal{A}}j$}\right\}\subset \mathcal{S}_p(b)\]
and let
\[\widehat{\mathbf{\Omega}}\in\argmin_{\mathbf{\Omega}\in\mathcal{S}_p(b;\mathcal{A})}F(\mathbf{\Omega})\;.\]
We will show that $\widehat{\mathbf{\Omega}}$ is a minimizer of $F(\mathbf{\Omega})$ over the larger set $\mathcal{S}_p(b)$.

Take any $\Delta=(\Delta^{(1)},\dots,\Delta^{(K)})$ with $\Delta^{(k)}\in\R^{d\times d}$ for each $k$. Let $D$ and $E$ be the block-diagonal and off-block-diagonal parts of $\Delta$; that is,
\[D^{(k)}_{ij}=\left\{\begin{array}{ll}\Delta^{(k)}_{ij}&\text{ if }i\sim_\mathcal{A} j\\0&\text{ if }i\not\sim_{\mathcal{A}}j\end{array}\right.\text{ and }E^{(k)}_{ij}=\left\{\begin{array}{ll}0&\text{ if }i\sim_\mathcal{A} j\\\Delta^{(k)}_{ij}&\text{ if }i\not\sim_{\mathcal{A}}j\end{array}\right.\;.\]
Suppose that $\widehat{\mathbf{\Omega}}+\Delta\in\mathcal{S}_p(b)$. Then 
\[b_k\geq \norm{\widehat\Omega^{(k)}+\Delta^{(k)}} \geq \norm{\left(\widehat\Omega^{(k)}+\Delta^{(k)}\right)_{A_r,A_r}} = \norm{\left(\widehat\Omega^{(k)}+D^{(k)}\right)_{A_r,A_r}}\;,\]
and so
\[\norm{\widehat{\Omega}^{(k)}+D^{(k)}}=\max_r  \norm{\left(\widehat\Omega^{(k)}+D^{(k)}\right)_{A_r,A_r}}\leq b_k\;,\]
proving that $\widehat{\mathbf{\Omega}}+D\in\mathcal{S}_p(b;\mathcal{A})$. Then
by optimality of $\widehat{\mathbf{\Omega}}$ over the set $\mathcal{S}_p(b;\mathcal{A})$,
\[F(\widehat{\mathbf{\Omega}}+D)\geq F(\widehat{\mathbf{\Omega}})\;.\]
Then
\begin{align*}
&F(\widehat{\mathbf{\Omega}}+\Delta)
 = F(\widehat{\mathbf{\Omega}}+D+E) \\
&=  \sum_{k=1}^K \frac{n_k}{2}\left(\inner{\widehat\Omega^{(k)}+D^{(k)}+E^{(k)}}{S^{(k)}}-\log\det(\widehat\Omega^{(k)}+D^{(k)}+E^{(k)})\right)\\
&\hspace{1in}+ \gamma\sum_{i<j}\beta\log\big(1+f(\widehat{\mathbf{\Omega}}_{ij}+D_{ij}+E_{ij})/\beta\big)\\
\intertext{Since $D_{ij}$ and $\widehat{\mathbf{\Omega}}_{ij}$ are nonzero only when $i,j$ are in the same block, and $E_{ij}$ is nonzero only if $i,j$ are in different blocks,}
&=  \sum_{k=1}^K \frac{n_k}{2}\left(\inner{\widehat\Omega^{(k)}+D^{(k)}}{S^{(k)}}-\log\det(\widehat\Omega^{(k)}+D^{(k)})\right) + \gamma\sum_{ i<j}\beta\log\big(1+f(\widehat{\mathbf{\Omega}}_{ij}+D_{ij})/\beta\big)\\
&\hspace{0.5in}+ \sum_{k=1}^K \frac{n_k}{2}\left(\inner{E^{(k)}}{S^{(k)}}-\log\det(\widehat\Omega^{(k)}+D^{(k)}+E^{(k)})+\log\det(\widehat\Omega^{(k)}+D^{(k)})\right) \\
&\hspace{1in}+ \gamma\sum_{i\not\sim_{\mathcal{A}} j, i<j}\beta\log\big(1+f(E_{ij})/\beta\big)\\
\intertext{Letting $M=\#\{(i,j):i\not\sim_{\mathcal{A}}j,i<j\}$,}
&=F(\widehat{\mathbf{\Omega}}+D)-M\gamma\cdot\beta\log(1+f(\zeros)/\beta)\\
&\hspace{0.5in} + \sum_{k=1}^K \frac{n_k}{2}\left(\inner{E^{(k)}}{S^{(k)}}-\log\det(\widehat\Omega^{(k)}+D^{(k)}+E^{(k)})+\log\det(\widehat\Omega^{(k)}+D^{(k)})\right) \\
&\hspace{1in}+ \gamma\sum_{i\not\sim_{\mathcal{A}} j, i<j}\beta\log\big(1+f(E_{ij})/\beta\big)\\
\intertext{By optimality of $\mathbf{\Omega}$ over $\mathcal{S}_p(b;\mathcal{A})$,}
&\geq F(\widehat{\mathbf{\Omega}}) -M\gamma\cdot\beta\log(1+f(\zeros)/\beta)\\
&\hspace{0.5in} + \sum_{k=1}^K \frac{n_k}{2}\left(\inner{E^{(k)}}{S^{(k)}}-\log\det(\widehat\Omega^{(k)}+D^{(k)}+E^{(k)})+\log\det(\widehat\Omega^{(k)}+D^{(k)})\right) \\
&\hspace{1in}+ \gamma\sum_{i\not\sim_{\mathcal{A}} j, i<j}\beta\log\big(1+f(E_{ij})/\beta\big)\\
\intertext{Since $\frac{\partial}{\partial \mathbf{\Omega}}\log\det(\mathbf{\Omega})=-\mathbf{\Omega}^{-1}$ and $\norm{\nabla^2\log\det(\mathbf{\Omega}')}$ is bounded for $\mathbf{\Omega}'$ near $\mathbf{\Omega}$, we can apply a Taylor expansion to the difference of $\log\det(\cdot)$ terms:}
&= F(\widehat{\mathbf{\Omega}})-M\gamma\cdot\beta\log(1+f(\zeros)/\beta)\\
&\hspace{0.5in} + \sum_{k=1}^K \frac{n_k}{2}\left(\inner{E^{(k)}}{S^{(k)}}-\inner{E^{(k)}}{-(\Omega^{(k)}+D^{(k)})^{-1}}-\O{\norm{E^{(k)}}^2_F}\right) \\
&\hspace{1in}+ \gamma\sum_{i\not\sim_{\mathcal{A}} j, i<j}\beta\log\big(1+f(E_{ij})/\beta\big)\\\intertext{Since $\widehat\Omega^{(k)}+D^{(k)}$ is block-diagonal and therefore so is $\left(\widehat\Omega^{(k)}+D^{(k)}\right)^{-1}$, while $E^{(k)}$ is supported off of the diagonal blocks,}
&= F(\widehat{\mathbf{\Omega}})-M\gamma\cdot\beta\log(1+f(\zeros)/\beta)\\
&\hspace{0.5in} + \sum_{k=1}^K \frac{n_k}{2}\left(\inner{E^{(k)}}{S^{(k)}}-\O{\norm{E^{(k)}}^2_F}\right) + \gamma\sum_{i\not\sim_{\mathcal{A}} j, i<j}\beta\log\big(1+f(E_{ij})/\beta\big)\\
\intertext{Applying Taylor expansion to the terms $\log(1+f(E_{ij})/\beta)$, for $E_{ij}$ sufficiently close to $0$,}
&= F(\widehat{\mathbf{\Omega}})-M\gamma\cdot\beta\log(1+f(\zeros)/\beta)\\
&\hspace{0.5in} + \sum_{k=1}^K \frac{n_k}{2}\left(\inner{E^{(k)}}{S^{(k)}}-\O{\norm{E^{(k)}}^2_F}\right)\\
&\hspace{1in} + \gamma\sum_{i\not\sim_{\mathcal{A}} j, i<j}\beta\left[\log(1+f(\zeros)/\beta)+\frac{f(E_{ij})-f(\zeros)}{\beta} - \O{\left(f(E_{ij})-f(\zeros)\right)^2}\right]\\
\intertext{Since $f$ is $L$-Lipschitz,}
&= F(\widehat{\mathbf{\Omega}})-M\gamma\cdot\beta\log(1+f(\zeros)/\beta)\\
&\hspace{0.5in} + \sum_{k=1}^K \frac{n_k}{2}\left(\inner{E^{(k)}}{S^{(k)}}-\O{\norm{E^{(k)}}^2_F}\right)\\
&\hspace{1in} + \gamma\sum_{i\not\sim_{\mathcal{A}} j, i<j}\beta\left[\log(1+f(\zeros)/\beta)+\frac{f(E_{ij})-f(\zeros)}{\beta} - L^2\cdot  \O{\norm{E_{ij}}^2_2}\right]\\
\intertext{Simplifying,}
&= F(\widehat{\mathbf{\Omega}})
 + \sum_{k=1}^K \frac{n_k}{2}\left(\inner{E^{(k)}}{S^{(k)}}-\O{\norm{E^{(k)}}^2_F}\right)\\
&\hspace{1in} + \gamma\sum_{i\not\sim_{\mathcal{A}} j, i<j}\beta\left[\frac{f(E_{ij})-f(\zeros)}{\beta} - L^2\cdot  
\O{\norm{E_{ij}}^2_2}\right]\\
\intertext{If we assume that $-\gamma^{-1}\cdot\diag\{n_1,\dots,n_K\} \cdot \mathbf{S}_{ij}\in\partial f(\zeros)$ for all $i\not\sim_{\mathcal{A}} j$,}
&\geq F(\widehat{\mathbf{\Omega}})+ \sum_{k=1}^K \frac{n_k}{2}\left(\inner{E^{(k)}}{S^{(k)}}-\O{\norm{E^{(k)}}^2_F}\right) \\
&\hspace{1in}+ \gamma\sum_{i\not\sim_{\mathcal{A}} j, i<j}\beta\left[\frac{\inner{E_{ij}}{-\gamma^{-1}\cdot\diag\{n_1,\dots,n_K\} \cdot \mathbf{S}_{ij}}}{\beta} - L^2\cdot  \O{\norm{E_{ij}}^2_2}\right]\\
\intertext{Cancelling out the terms that are linear in $E$,}
&= F(\widehat{\mathbf{\Omega}})- \sum_{k=1}^K \frac{n_k}{2}\O{\norm{E^{(k)}}^2_F} - \alpha\sum_{i\not\sim_{\mathcal{A}} j, i<j} L^2\cdot  \O{\norm{E_{ij}}^2_2}\\
&=F(\widehat{\mathbf{\Omega}}) -\O{ \sum_{i,j,k}E^{(k)}_{ij}{}^2}
\geq F(\widehat{\mathbf{\Omega}}) -\O{ \sum_{i,j,k}\Delta^{(k)}_{ij}{}^2}\;.\end{align*}
 Since $F$ is convex over $\mathcal{S}_p(b)$ which is itself a convex set, this is sufficient to prove that $\widehat{\mathbf{\Omega}}$ is a minimizer of $F(\cdot)$ over $\mathcal{S}_p(b)$. 
\end{proof}

\end{document}